\documentclass[prodmode,permissions]{acmsmall-ec16}

\usepackage[normalem]{ulem}
\usepackage{hyperref}
\usepackage[ruled]{algorithm2e}
\usepackage[toc,page]{appendix}
\renewcommand{\algorithmcfname}{ALGORITHM}
\usepackage{amsmath,amsfonts,amssymb,bbm} 
\usepackage[numbers,sort&compress]{natbib} 
\SetArgSty{textrm}  
\SetAlFnt{\small}
\SetAlCapFnt{\small}
\SetAlCapNameFnt{\small}
\SetAlCapHSkip{0pt}
\IncMargin{-\parindent}

\doi{XXXXXXX.XXXXXXX}

\newcommand{\ind}{\mathbbm{1}}
\newcommand{\pr}{\mathbb{P}_{\mathcal{T}}}
\newtheorem{question}{Open Question}

\begin{document}


\title{Condorcet-Consistent and Approximately Strategyproof Tournament Rules} 
\author{JON SCHNEIDER \footnote{Department of Computer Science, Princeton University, js44@cs.princeton.edu.}
\affil{Princeton University}
ARIEL SCHVARTZMAN  \footnote{Department of Computer Science, Princeton University, acohenca@cs.princeton.edu.}
\affil{Princeton University}
S. MATTHEW WEINBERG \footnote{Department of Computer Science, Princeton University, smweinberg@princeton.edu.}
\affil{Princeton University}}


\begin{abstract} 
We consider the manipulability of tournament rules for round-robin tournaments of $n$ competitors. Specifically, $n$ competitors are competing for a prize, and a tournament rule $r$ maps the result of all $\binom{n}{2}$ pairwise matches (called a \emph{tournament}, $T$) to a distribution over winners. Rule $r$ is \emph{Condorcet-consistent} if whenever $i$ wins all $n-1$ of her matches, $r$ selects $i$ with probability $1$. 

We consider strategic manipulation of tournaments where player $j$ might throw their match to player $i$ in order to increase the likelihood that one of them wins the tournament. Regardless of the reason why $j$ chooses to do this, the potential for manipulation exists as long as $\Pr[r(T) = i]$ increases by more than $\Pr[r(T) = j]$ decreases. Unfortunately, it is known that every Condorcet-consistent rule is manipulable~\cite{AltmanKleinberg}. In this work, we address the question of \emph{how manipulable} Condorcet-consistent rules must necessarily be - by trying to minimize the difference between the increase in $\Pr[r(T) = i]$ and decrease in $\Pr[r(T) = j]$ for any potential manipulating pair.


We show that every Condorcet-consistent rule is in fact $1/3$-manipulable, and that selecting a winner according to a random single elimination bracket is not $\alpha$-manipulable for any $\alpha > 1/3$. We also show that many previously studied tournament formats are all $1/2$-manipulable, and the popular class of Copeland rules (any rule that selects a player with the most wins) are all in fact $1$-manipulable, the worst possible. Finally, we consider extensions to match-fixing among sets of more than two players.
\end{abstract}
   
\maketitle

\section{Introduction}
In recent years, numerous scandals have unfolded surrounding match fixing and throwing at the highest levels of competitive sports (e.g. Olympic Badminton~\cite{badminton:bbc}, Professional Tennis~\cite{tennis:bbc}, European Football~\cite{eurofootball:bbc}, and even eSports~\cite{esports:SC2}). In some instances, the motivation behind these scandals was gambling profits, and no amount of clever tournament design can possibly mitigate this. In others, however, the surprising motivation was an improved performance \emph{at that same tournament.} For instance, four Badminton teams (eight players) were disqualified from the London 2012 Olympics for throwing matches. Interestingly, the reason teams wanted to lose their matches was in order to \emph{improve} their probability of winning an Olympic medal. Olympic Badminton (like many other sports) conducts a two-phase tournament. In the first stage, groups of four play a round-robin tournament, with the top two teams advancing. In the second stage, the advancing teams participate in a single elimination tournament, seeded according to their performance in the group stage. An upset in one group left one of the world's top teams with a low seed, so many teams actually preferred to receive a \emph{lower} seed coming out of the group stage to face the tougher opponent as late as possible.

While much of the world blames the teams for their poor sportsmanship, researchers in voting theory have instead critiqued the poor tournament design that punished teams for trying to maximize their chances of winning a medal. Specifically, the two-phase tournament lacks the basic property of \emph{monotonicity}, where no competitor can unilaterally improve their chances of winning by throwing a match that they otherwise could have won. Thus, recent work has addressed the question of whether tournament structures exist that are both fair, in that they select some notion of a qualified winner, and strategyproof, in that teams have no incentive to do anything but play their best in each match.

One minimal notion of fairness studied is \emph{Condorcet-consistence}, which just guarantees that whenever one competitor wins \emph{all} of their matches (and is what's called a \emph{Condorcet winner}), they win the event with probability $1$. Designing Condorcet-consistent, monotone rules is simple: any single elimination bracket suffices. Popular voting rules such as the Copeland Rule or the Random Condorcet Removal Rule are also Condorcet-consistent and monotone, but two-phase tournaments with an initial group play aren't~\cite{Pauly:Strategizing}.

Still, monotonicity only guarantees that no team wishes to unilaterally throw a match to improve their chances of winning, whereas one might also hope to guarantee that no two teams could fix the outcome of their match in order to improve the probability that one of them wins. While we have to go back further in history to find a clear instance of this kind of match-fixing, it did indeed result in a historical scandal. In the 1982 FIFA World Cup (again a two-stage tournament), Austria, West Germany, and Algeria were in the same group of four where two would advance. Algeria had already won two matches and lost one, Austria was 2-0, West Germany was 1-1, and the only remaining game was Austria vs. West Germany. Due to tie-breakers and the specific outcomes of previous matches, Austria would have been eliminated by a large West German victory, and West Germany would have been eliminated by a loss or draw. Once West Germany scored an early goal, \emph{both} teams essentially threw the rest of the match, allowing both of them to advance at Algeria's expense~\cite{eurofootball:guardian}. While the incident was never formally investigated, many fans were confident the teams had colluded beforehand, and the event is remembered as the ``disgrace of Gij\'{o}n.'' Before being eliminated, Algeria had become the first African team to beat a European team at the World Cup, and also the first to win two games. West Germany went on to become the runners-up of the tournament. 

Motivated by events like this, it is important also to design tournaments where no two teams can fix the outcome of their match and improve the probability that one of them wins. Altman and Kleinberg terms this property $2$-Strongly Nonmanipulable ($2$-SNM), and showed that no tournament rule is both Condorcet-consistent and $2$-SNM~\cite{AltmanKleinberg} (it was previously shown by Altman et. al. that no \emph{deterministic} rule is both Condorcet-consistent and $2$-SNM~\cite{Altman:Nonmanipulable}). 

In light of this, both works relax the notion of Condorcet-consistency and design tournament rules that are at least \emph{non-imposing} (could possibly select each competitor as a winner) and $2$-SNM~\cite{Altman:Nonmanipulable}, or $\alpha$-Condorcet-consistent (if there is a Condorcet winner, she wins with probability at least $\alpha$) and $2$-SNM. While these relaxations are well-motivated for settings where pair-wise comparisons are only \emph{implicitly} made, and not even necessarily learned in the end (e.g. elections), it is hard to imagine a successful sports competition format where a competitor could win all their matches and still leave empty handed. 

Motivated by match-based applications such as sporting events, where the outcome of pair-wise matches is \emph{explicitly} learned and used to select a winner, we consider instead the design of tournament rules that are exactly Condorcet-consistent, but only approximately $2$-SNM. Specifically, we say that a tournament rule is $2$-SNM-$\alpha$ if it is \emph{never} possible for two teams $i$ and $j$ to fix their match such that the probability that the winner is in $\{i,j\}$ improves by at least $\alpha$. The idea behind this relaxation is that whatever motivates $j$ to throw the match (perhaps $j$ and $i$ are teammates, perhaps $i$ is paying $j$ some monetary bribe, etc.), the potential gains scale with $\alpha$. So it is easier to disincentivize manipulation (either through investigations and punishments, reputation, or just feeling morally lousy) in tournaments that are less manipulable.


\subsection{Our Results}
Our main result is a matching upper and lower bound of $1/3$ on attainable values of $\alpha$ for Condorcet-consistent $2$-SNM-$\alpha$ tournament rules. The optimal rule that attains this upper bound is actually quite simple: a random single elimination bracket. Specifically, each competitor is randomly placed into one of $2^{\lceil \log_2 n \rceil}$ seeds, along with $2^{\lceil \log_2 n \rceil} - n$ byes, and then a single elimination tournament is played. 

Proving a lower bound of $1/3$ is straight-forward: imagine a tournament with three players, $A, B$ and $C$, where $A$ beats $B$, $B$ beats $C$, and $C$ beats $A$. Then some pair must win with combined probability at most $2/3$. Yet, any pair could create a Condorcet winner by colluding, who necessarily wins with probability $1$ in any Condorcet-consistent rule. Embedding this within examples for arbitrary $n$ is also easy: just have $A$, $B$, and $C$ each beat all of the remaining $n-3$ competitors\footnote{{Interestingly, this lower-bound example is far from pathological and occurs at even the highest levels of professional sports (see \cite{tennis:grantland}, for instance).}}. 

On the other hand, proving that a random single elimination bracket is optimal is tricky, but our proof is still rather clean. For any $i, j$ in any tournament, we directly show that $i$ can improve her probability of winning by at most $1/3$ when $j$ throws their match using a coupling argument. For every deterministic single elimination bracket where $i$ and $j$ could potentially gain from manipulation (because $i$ would be the champion if $i$ beat $j$, but $j$ would \emph{not} be the champion even if $j$ beat $i$), we construct \emph{two} deterministic single elimination brackets where no potential exists (possibly because one of them will lose before facing each other, or because the winner would be in $\{i,j\}$ no matter the outcome of their match). For our coupling to be valid, we not only need each mapping to be invertible, but also for their images to be disjoint. Our coupling is necessarily somewhat involved in order to obtain this property, but otherwise we believe our proof is likely as simple as possible. Because the probability that $j$ wins cannot possibly go up by throwing a match to $i$, this immediately proves that a random single elimination bracket is $2$-SNM-$1/3$. 

{We also show that the Copeland rule, a popular rule that chooses the team with the most wins, is asymptotically $2$-SNM-$1$, the \emph{worst} possible. Essentially, the problem is that if all teams have the same number of wins, then any two can collude to guarantee that one of them wins, no matter the tie-breaking rule. We further show that numerous other formats, (the Random Voting Caterpillar, the Iterative Condorcet Rule, and the Top Cycle Rule) are all at best $2$-SNM-$1/2$. The same example is bad for all three formats: there is one superman who beats $n-2$ of the remaining players, and one kryptonite, who beats only the superman (but loses to the other $n-2$ players).}

Our results extend to settings where the winner of each pairwise match is not deterministically known, but randomized (i.e. all partipants know that $i$ will beat $j$ with probability $p_{ij}$). Specifically, we show that any rule that is $2$-SNM-$\alpha$ when all $p_{ij} \in \{0,1\}$ is also $2$-SNM-$\alpha$ for arbitrary $p_{ij}$. Clearly, any lower bound using integral $p_{ij}$ also provides a lower bound for arbitrary $p_{ij}$, so as far as upper/lower bounds are concerned the models are equivalent. Of course, the randomized model is much more realistic, so it is convenient that we can prove theorems in this setting by only studying the deterministic setting, which is mathematically much simpler.

Finally, we consider manipulations among coalitions of $k > 2$ participants. We say that a rule is $k$-SNM-$\alpha$ if no set $S$ of size $\leq k$ can \emph{ever} manipulate the outcomes of matches between players in $S$ to improve the probability that the winner is in $S$ by more than $\alpha$. We prove a simple lower bound of $\alpha = \frac{k-1}{2k-1}$ on all Condorcet-consistent rules, and conjecture that this is tight. 


\subsection{Related Works}

The mathematical study of tournament design has a rich literature, ranging from social choice theory to psychology. The overarching goal in these works is to design tournament rules that satisfy various properties a designer might find desirable. Examples of such properties might be that all players are treated equally, that a winner is chosen without a tiebreaking procedure, or that a ``most qualified'' winner is selected ~\cite{Fishburn, Rubinstein, Dutta, Rivest:GT, Young, Moulin, Schwartz}. See~\cite{Laslier1997} for a good review of this literature and its connections to other fields as well.

Most related to our work are properties involving \emph{strategic manipulation}. In the more general field of Voting Theory, there is a rich literature on the design of strategyproof mechanisms dating back to Arrow's Impossibility Theorem~\cite{Arrow} and the Gibbard-Satterthwaite Theorem~\cite{Gibbard,Satterthwaite,Gibbard2}. While tournaments are a very special case (voters are indifferent among outcomes where they do not win, voters can only ``lie'' in specific ways, etc.), tournament design indeed seems to inherit much of the impossibility associated with strategyproof voting procedures~\cite{AltmanKleinberg}, ~\cite{Altman:Nonmanipulable}. 

Specifically, Altman et. al. proved that no deterministic tournament rule is $2$-SNM and Condorcet-consistent, and Altman and Kleinberg proved that no randomized tournament rule is $2$-SNM and Condorcet-consistent either~\cite{Altman:Nonmanipulable, AltmanKleinberg}. More recently, Pauly studied the specific two-stage tournament rule used by the World Cup (and Olympic Badminton, etc.)~\cite{Pauly:Strategizing}. There, it is shown essentially that the problem lies in the first round group stage: no changes to the second phase can possibly result in a strategyproof \footnote{See~\cite{Pauly:Strategizing} for the specific notion of strategyproofness studied.} tournament.

To cope with their impossibility results, Altman et. al. propose a relaxation of Condorcet-consistence called \emph{non-imposing}. A rule $r$ is non-imposing if for all $i$, there exists a $T$ such that player $i$ wins with probability 1. They design a clever recursive rule that is non-imposing and $2$-SNM for all $n \neq 3$. Interestingly, they also show that for $n=3$ no such rule exists. Altman and Kleinberg consider a different relaxation called \emph{$\alpha$-Condorcet-consistent}. A rule $r$ is $\alpha$-Condorcet-consistent if whenever $i$ is a Condorcet winner in $T$, we have their probability of winning $T$ is at least $\alpha$. They design a rule that is $2/n$-Condorcet-consistent and $2$-SNM (in fact it is also $k$-SNM for all $k$), but conjecture that much better is attainable. 

The two works above are most similar to ours in spirit: motivated by the non-existence of Condorcet-consistent and $2$-SNM tournament rules, we relax one of the notions. These previous works relax Condorcet-consistency while maintaining $2$-SNM exactly, and are most appropriate in settings where pairwise comparisons of players are only learned \emph{implicitly} (or perhaps not at all) through the outcome and not \emph{explicitly} as the result of matches. Instead, we relax the notion of $2$-SNM and maintain the notion of Condorcet-consistency exactly. In settings like sports competitions where pairwise comparisons of players are learned explicitly through matches played, Condorcet-consistency is a non-negotiable desideratum. Therefore, we believe our approach is more natural in such settings. 

Several authors have taken alternate approaches to understanding the power of manipulation in tournament design. A recent line of work  \cite{stanton2011manipulating, aziz2014fixing, kim2015can} studies manipulation from the perspective of the tournament organizer. More specifically, they study the computational complexity of the problem of \emph{rigging} a tournament: given a collection of pairwise outcomes between players, which players can possibly win a single-elimination tournament (and how should the tournament be seeded so that a specific player wins)? Kern and Paulusma consider a similar question for round-robin tournaments where some matches have been already been played \cite{kern2004computational}. 

A large body of literature exists regarding manipulation and bribery in the context of voting rules. For an introduction, we recommend the reader consult chapter 7 of the handbook \cite{moulin2016handbook}. 

\subsection{Conclusions and Future Work}
Our work contributes to a recent literature on incentive compatible tournament design. While most previous works insist on strong incentive properties and relaxed fairness properties, such rules are inadequate for sporting events. Instead, we insist at least that events maintain Condorcet-consistency, and aim to relax strategyproofness as minimally as possible. 

{At a high level, our work suggests (similar to previous works), that single elimination brackets are desirable whenever incentive issues come into play. However, previous desiderata (such as those considered in~\cite{AltmanKleinberg}) don't necessarily rule out other tournament formats, like the Copeland rule, which is ubiquitous in tournaments (both as a complete format and as subtournaments in a two-phase format). In comparison, our work identifies single elimination brackets ($2$-SNM-$1/3$) as having significantly better strategic properties versus the Copeland rule ($2$-SNM-$1$).  }

Our work also identifies two practical suggestions when match-fixing is a concern that aren't explained by prior benchmarks. First, when hosting a single elimination tournament, it might be desirable to release the exact bracket as late as possible. The idea is that as soon as the exact bracket is known, competitors have greater incentive to fix matches (in our model, up to three times as much), which presumably takes some time and organization. Obviously, there are more tradeoffs at play: a later release inconveniences athletes and fans, and (perhaps more importantly to the designers) could negatively impact ticket sales. But our work does at least identify match-fixing as a part of this tradeoff. Note that some Olympic events (such as Taekwondo) contest the entire competition in a single day at a single venue, so a delayed release may indeed be practical. We also note that a similar ``fix'' was applied after the 1982 World Cup: the last two matches in each group are now played at the same time to minimize the amount of information teams have when making potentially strategic decisions.

Additionally, our work suggests that even in the optimal tournament, hefty punishments for cheaters might be necessary in order to discourage match-fixing (even without taking gambling into consideration). In many sports, winning an Olympic gold can make a career. Unfortunately, our work suggests that punishments roughly on this order might be necessary in order to properly deter match-fixing.

Finally, we propose two directions for future work. First, while we obtain tight results for Condorcet-consistent $2$-SNM-$\alpha$ rules, we only prove a lower bound of $k$-SNM-$\frac{k-1}{2k-1}$ for Condorcet-consistent rules and $k > 2$. We conjecture that this is tight, but unfortunately simulations indicate that all of the formats studied in our work do \emph{not} achieve this bound. So it is an interesting open question to design a rule that does. Even partial results (of the form identified below) would require a new tournament format than those considered in this work.

\begin{question}
Does there exist a tournament rule that is Condorcet-consistent and $k$-SNM-$\frac{k-1}{2k-1}$ for all $k$? What about a family of rules $\mathcal{F}$ such that for all $k$, $F_k$ is $k$-SNM-$\frac{k-1}{2k-1}$? What about a rule that is $k$-SNM-$1/2$ for all $k$?\footnote{Note that $\frac{k-1}{2k-1} \rightarrow 1/2$ as $k \rightarrow \infty$.}
\end{question}

It is also important to study what bounds are attainable in restricted versions of our probabilistic model (e.g. if for all $i, j$, the probability that $i$ beats $j$ lies in $[\epsilon, 1-\epsilon]$). Realistic instances at least have \emph{some} non-zero probability of an upset in every match, but our lower bounds don't hold in this model. So it is interesting to see if better formats are possible. 

\begin{question}
Is a random single elimination bracket still optimal among Condorcet-consistent rules (w.r.t. $2$-SNM-$\alpha$) if for all $i, j$, the probability that $i$ beats $j$ lies in $[\epsilon, 1-\epsilon]$? How does the optimal attainable $\alpha$ for Condorcet-consistent, $2$-SNM-$\alpha$ tournament formats change as a function of $\epsilon$?
\end{question}

\section{Preliminaries and Notation}

In this section, we present notation used throughout the remainder of the paper. Where possible, we adopt notation from~\cite{AltmanKleinberg}.

\begin{definition}
A (round-robin) \textit{tournament} $T$ on $n$ players is the set of outcomes of the $\binom{n}{2}$ games played between all pairs of distinct players. We write $T_{ij} = 1$ if player $i$ beats player $j$ and $T_{ij} = -1$ otherwise. We also let $\mathcal{T}_n$ denote the set of tournaments on $n$ players.
\end{definition}

\begin{definition}
For a subset $S \subseteq [n]$ of players, two tournaments $T$ and $T'$ are \textit{$S$-adjacent} if they only differ on the outcomes of some subset of games played between members of $S$. In particular, two tournaments $T$ and $T'$ are $\{i, j\}$ adjacent if they only differ in the result of the game played between player $i$ and player $j$.
\end{definition}

\begin{definition}
A \textit{tournament rule} {(or \textit{winner determination rule})} $r:\mathcal{T}_n \rightarrow \Delta([n])$ is a mapping from the set of tournaments on $n$ players to probability distributions over these $n$ players (representing the probability we choose a given player to be the winner). We will write $r_i(T) = \mathrm{Pr}[r(T)=i]$ to denote the probability that player $i$ wins tournament $T$ under rule $r$.
\end{definition}

Many tournament rules, while valid by the above definition, would be ill-suited for running an actual tournament; for example, the tournament rule which always crowns player 1 the winner. In an attempt to restrict ourselves to `reasonable' tournament rules, we consider tournaments that obey the following two criteria.

\begin{definition}
Player $i$ is a \textit{Condorcet winner} in tournament $T$ if player $i$ wins their match against all the other $n-1$ players. A tournament rule $r$ is \textit{Condorcet-consistent} if $r_i(T) = 1$ whenever $i$ is a Condorcet winner in $T$.
\end{definition}

\begin{definition}
A tournament rule $r$ is \textit{monotone} if, for all $i$, $r_i(T)$ does not increase when $i$ loses a game it wins in $T$. That is, if $i$ beats $j$ in $T$ and $T$ and $T'$ are $\{i, j\}$ adjacent, then if $r$ is monotone, $r_i(T) \geq r_i(T')$.
\end{definition}

Intuitively, this first criterion requires us to award the prize to the winner in the case of a clear winner (hence making the tournament a contest of skill), and the second criterion makes it so that players have an incentive to win their games. There are various other criteria one might wish a tournament rule to satisfy; many can be found in~\cite{AltmanKleinberg}.

In this paper, we consider the scenario where certain coalitions of players attempt to increase the overall chance of one of them winning by manipulating the outcomes of matches within players of the coalition. The simplest case of this is in the case of coalitions of size 2, where player $j$ might throw their match to player $i$. If $T$ is the original tournament and $T'$ is the manipulated tournament where $j$ loses to $i$, then player $i$ gains $r_i(T') - r_i(T)$ from the manipulation, and player $j$ loses $r_j(T) - r_j(T')$ (in terms of probability of winning). Therefore, as long as $r_i(T') - r_i(T) > r_j(T) - r_j(T')$, it will be in the players' interest to manipulate. Equivalently, if $r_i(T') + r_j(T') > r_i(T) + r_j(T)$ (i.e., the probability either player $i$ or $j$ wins increases upon throwing the match), there is incentive for $i$ and $j$ to manipulate.

Ideally, we would like to choose a tournament rule so that, regardless of the tournament, there will be no incentive to perform manipulations of the above sort. This is encapsulated in the following definition from~\cite{AltmanKleinberg}.

\begin{definition}
A tournament rule $r$ is \textit{2-strongly non-manipulable (2-SNM)} if, for all pairs of $\{i, j\}$-adjacent tournaments $T$ and $T'$, $r_i(T) + r_j(T) = r_i(T') + r_j(T')$.
\end{definition}

Unfortunately, no tournament rules exist that are simultaneously Condorcet-consistent and 2-strongly non-manipulable (this is shown in~\cite{AltmanKleinberg} and also follows from our lower bound in Section~\ref{sec:LB}). As tournament designers, one way around this obstacle is to discourage manipulation. This discouragement can take many forms, both explicit (if players are caught fixing matches, they are disqualified/fined) and implicit (it is logistically hard to fix matches, it is unsportsmanlike). So the focus of this paper is to quantify \emph{how manipulable} certain tournament formats are (i.e. how much can teams possibly gain by fixing matches), the idea being that it is easier to discourage manipulation in tournaments that are less manipulable.


\begin{definition}
A tournament rule $r$ is \textit{2-strongly non-manipulable at probability $\alpha$ (2-SNM-$\alpha$)} if, for all $i$ and $j$ and pairs of $\{i, j\}$-adjacent tournaments $T$ and $T'$, $r_i(T') + r_j(T') - r_i(T) - r_j(T) \leq \alpha$.
\end{definition}

It is straightforward to generalize this definition to larger coalitions of colluding players.

\begin{definition}
A tournament rule $r$ is \textit{$k$-strongly non-manipulable at probability $\alpha$ ($k$-SNM-$\alpha$)} if, for all subsets $S$ of players of size at most $k$, for all pairs of $S$-adjacent tournaments $T$ and $T'$, $\sum_{i\in S}r_i(T') - \sum_{i\in S}r_i(T) \leq \alpha$.
\end{definition}

\subsection{The Random Single-Elimination Bracket Rule}

Our main result concerns a specific tournament rule we call the \textit{random single-elimination bracket rule}. This rule can be defined formally as follows.

\begin{definition}
A \textit{single-elimination bracket} (or \textit{bracket}, for short) $B$ on $n = 2^{h}$ players is a complete binary tree of height $h$ whose leaves are labelled with some permutation of the $n$ players. The outcome of a bracket $B$ under a tournament $T$ is the labelling of internal nodes of $B$ where each node is labelled by the winner of its two children under $T$. The winner of $B$ under $T$ is the label of the root of $B$ under this labelling.
\end{definition}
\begin{definition}
The \textit{random single-elimination bracket rule} $r$ is a tournament rule on $n = 2^{h}$ players where $r_i(T)$ is the probability player $i$ is the winner of $B$ under $T$ when $B$ is chosen uniformly at random from the set of $n!$ possible brackets. 

If $n$ is not a power of $2$, we define the random single-elimination bracket rule on $n$ players by introducing $2^{\lceil \log_2 n \rceil} - n$ dummy players who lose to all of the existing $n$ players. 
\end{definition}

It is straightforward to check that the random single-elimination bracket rule is both Condorcet-consistent and monotone. Our main result (Theorem~\ref{thm:main}) shows that in addition to these properties, the random single-elimination bracket rule is $2$-SNM-$1/3$ (which is the best possible, by Theorem~\ref{thm:lbnd}). 

We give some examples of other common tournament rules in Section~\ref{sect:otherrules}. While  many of these rules are both Condorcet-consistent and monotone, we do not know of any which are additionally $2$-SNM-$1/3$.


\section{Main Result}

\subsection{Lower bounds for $k$-SNM-$\alpha$}
\label{sec:LB}
We begin by showing that no tournament rule is $2$-SNM-$\alpha$ for $\alpha < 1/3$. A similar theorem appears as Proposition~17 in~\cite{AltmanKleinberg} (which states that $\alpha = 0$ is impossible).

\begin{theorem}\label{thm:lbnd}
There is no Condorcet-consistent tournament rule on $n$ players (for $n \geq 3$) that is $2$-SNM-$\alpha$ for $\alpha < \frac{1}{3}$. 
\end{theorem}
\begin{proof}
Consider the tournament $T$ on three players $A$, $B$, and $C$ where $A$ beats $B$, $B$ beats $C$, and $C$ beats $A$ (illustrated in Figure~\ref{fig:lowerbound}).  . Note that, while this tournament has no Condorcet winner, changing the result of any of the three games results in a Condorcet winner. For example, if $A$ bribes $C$ to lose to $A$, then $A$ becomes the Condorcet winner.

If we have a tournament rule $r$ that is $2$-SNM-$\alpha$, then combining this with the above fact gives rise to the following three inequalities.

\begin{eqnarray*}
r_A(T) + r_B(T) &\geq & 1 - \alpha \\
r_B(T) + r_C(T) &\geq & 1 - \alpha \\
r_C(T) + r_A(T) &\geq & 1 - \alpha
\end{eqnarray*}

Together these imply $r_A(T) + r_B(T) + r_C(T) \geq \frac{3}{2}(1-\alpha)$. But $r_{A}(T) + r_B(T) + r_C(T) = 1$; it follows that $\alpha \geq \frac{1}{3}$, as desired.

We can extend this counterexample to $n > 3$ players by introducing $n-3$ dummy players who all lose to $A$, $B$, and $C$; the argument above continues to hold. 

\begin{figure*}
\centering
\includegraphics[scale = 0.25]{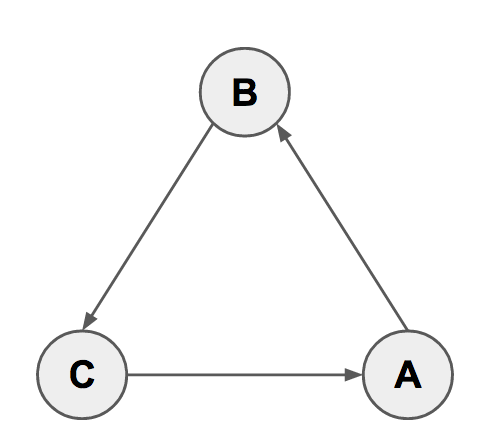}
\caption{A tournament which attains the lower bound of $\alpha=1/3$ for all tournament rules.}
\label{fig:lowerbound}
\end{figure*}
\end{proof}

We can use similar logic to prove lower bounds for the more general case of $k$-SNM-$\alpha$. 

\begin{theorem}\label{thm:lbndK}
There is no Condorcet-consistent tournament rule on $n$ players (for $n \geq 2k-1$) that is $k$-SNM-$\alpha$ for $\alpha < \frac{k-1}{2k-1}$.
\end{theorem}
\begin{proof}
Consider the following tournament $T$ on the $2k-1$ players labelled $1$ through $2k-1$. Each player $i$ wins their match versus the $k-1$ players $i+1, i+2, \dots, i+(k-1)$, and loses their match versus the $k-1$ players $i-1, i-2, \dots, i-(k-1)$ (indices taken modulo $2k-1$). Note that the coalition of players $S_i = \{i, i-1, \dots, i-(k-1)\}$ of size $k$ can cause $i$ to become a Condorcet winner if all players in the coalition agree to lose their games with $i$. If we have a tournament rule $r$ that is $k$-SNM-$\alpha$, then this implies the following $2k-1$ inequalities (one for each $i \in [2k-1]$):

\begin{equation}
\sum_{j \in S_i} r_j(T) \geq 1 - \alpha
\end{equation}

\noindent
Summing these $2k-1$ inequalities, we obtain

\begin{equation}
k\sum_{j=1}^{2k-1} r_j(T) \geq (2k-1)(1-\alpha)
\end{equation}

Since $\sum_{j=1}^{2k-1} r_j(T) \leq 1$, this implies that $\alpha \geq \frac{k-1}{2k-1}$, as desired. Again, it is possible to extend this example to any number of players $n \geq 2k-1$ by introducing dummy players who lose to all $2k-1$ of the above players.
\end{proof}

\subsection{Random single elimination brackets are 2-SNM-$1/3$}\label{sec:main}

We now show that the random single elimination bracket rule is optimal against coalitions of size $2$. The proof idea is simple; for every bracket $B$ that contributes to the incentive to manipulate $r_i(T') + r_j(T') - r_i(T) - r_j(T)$ we will show that there are two that do not (in other words, for every scenario where team $i$ benefits from the manipulation, there exist two other scenarios where the maniuplation does not benefit either team). 

\begin{theorem}\label{thm:main}
The random single elimination bracket rule is $2$-SNM-$1/3$. 
\end{theorem}

\begin{proof}
Let $\mathcal{B}$ be the set of $n!$ different possible brackets amongst the $n$ players. For a given tournament $T$ and a given player $i$, write $\ind(B,T,i)$ to represent the indicator variable which is $1$ if $i$ wins bracket $B$ under the outcomes in $T$ and $0$ otherwise. Then we can write

\begin{equation*}
r_i(T) = \frac{1}{|\mathcal{B}|}\sum_{B \in \mathcal{B}} \ind(B, T, i).
\end{equation*}

Assume $i$ loses to $j$ in $T$. Then, if we let $T'$ be the tournament that is $\{i, j\}$ adjacent to $T$, we can write the increase in utility resulting from $j$ throwing to $i$

\begin{equation}\label{eq:indgap}
\frac{1}{|\mathcal{B}|}\sum_{B \in \mathcal{B}} \left(\ind(B, T', i) + \ind(B, T', j) - \ind(B, T, i) - \ind(B, T, j)\right).
\end{equation}

Our goal is to show that this sum is at most $1/3$. Now, note that if $i$ does not end up playing $j$ in bracket $B$ under $T$, $i$ also does not play $j$ in $B$ under $T'$ (and vice versa). In these brackets, $\ind(B, T', i) = \ind(B, T, i)$ and $\ind(B, T', j) = \ind(B, T, j)$, so these brackets contribute nothing to the sum in Equation \ref{eq:indgap}. On the other hand, in a bracket $B$ where $i$ does play $j$, we are guaranteed that $\ind(B, T, i) = 0$ and $\ind(B, T', j) = 0$ (since $i$ loses to $j$ in $T$ and $j$ loses to $i$ in $T'$). Therefore, letting $\mathcal{B}_{ij}$ be the subset of $\mathcal{B}$ of brackets where $i$ meets $j$, we can rewrite Equation \ref{eq:indgap} as

\begin{equation*}
\frac{1}{|\mathcal{B}|}\sum_{B \in \mathcal{B}_{ij}} \left(\ind(B, T', i) - \ind(B, T, j)\right).
\end{equation*}

\noindent
Since $\ind(B, T', i) \leq 1$, this is at most

\begin{equation*}
\frac{1}{|\mathcal{B}|}\sum_{B \in \mathcal{B}_{ij}} \left(1 - \ind(B, T, j)\right).
\end{equation*}

This final sum counts exactly the number of brackets $B$ where $i$ and $j$ meet (under $T$, so $j$ beats $i$) but $j$ does not win the tournament. Call such brackets \textit{bad}, and call the remaining brackets \textit{good}. We will exhibit two injective mappings $\sigma_i$ and $\sigma_j$ from bad brackets to good brackets such that the ranges of $\sigma_i$ and $\sigma_j$ are disjoint. This implies that there are at least twice as many good brackets as bad brackets, and thus that the sum above is at most $1/3$, completing the proof.

For both mappings, we will need the following terminology. Consider a bad bracket $B$, and consider the path from $j$ up to the root of this tree. The nodes of this path are labelled by players that $j$ would face if they got that far. More specifically, $j$ has some opponent in the first round. Should $j$ win, $j$ would face some opponent in the second round, then the third round, etc. all the way to the finals, and these opponents do not depend on the outcomes of any of $j$'s matches. Then since $B$ is a bad bracket, $j$ does not win, and at least one of the players on this path can beat $j$. Choose the \textbf{latest} such player (i.e. the closest to the root) and call this player $k$. Note that $k$ might \emph{not} be the player that knocks $j$ out of the tournament (that is the \emph{first} player along this path who would beat $j$). 

Suppose that $i$ and $j$ meet at height $h$ of the bracket (i.e. in the $h^{th}$ round). Let $B_i, B_j, B_k$ be the subtrees of height $h$ that contain $i$, $j$, and $k$ respectively. An example is shown in Figure ~\ref{fig:badinstance}. 

\begin{figure*}
\centering
\includegraphics[scale = 0.25]{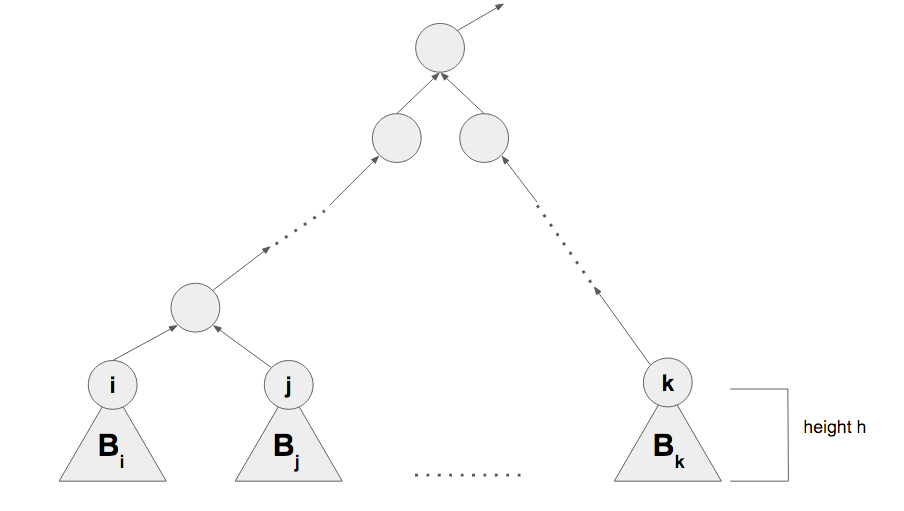}
\caption{An example of a bad bracket $B$.}
\label{fig:badinstance}
\end{figure*}

We first describe the simpler of the two maps, $\sigma_i$. Define $\sigma_i(B)$ by swapping the subtrees $B_i$ and $B_k$ as shown in Figure ~\ref{fig:goodinstance1}. In this bracket $j$ will lose to $k$ before ever meeting $i$, so $\sigma_i(B)$ is good. Moreover $\sigma_i$ is injective since we can construct its inverse. In $\sigma_i(B)$, $j$ certainly would lose to $k$ at height $h$ before reaching $i$. Furthermore, because we didn't change $B_j$ at all, $j$ still wins all of its first $h-1$ matches and makes it to $k$ (because we started from a $B$ where $j$ makes it to $i$ at height $h$). So we can identify $k$ as the first player who beats $j$ in $\sigma_i(B)$, learn the height $h$, and undo the swap of $B_k$ and $B_i$.

\begin{figure*}
\centering
\includegraphics[scale = 0.25]{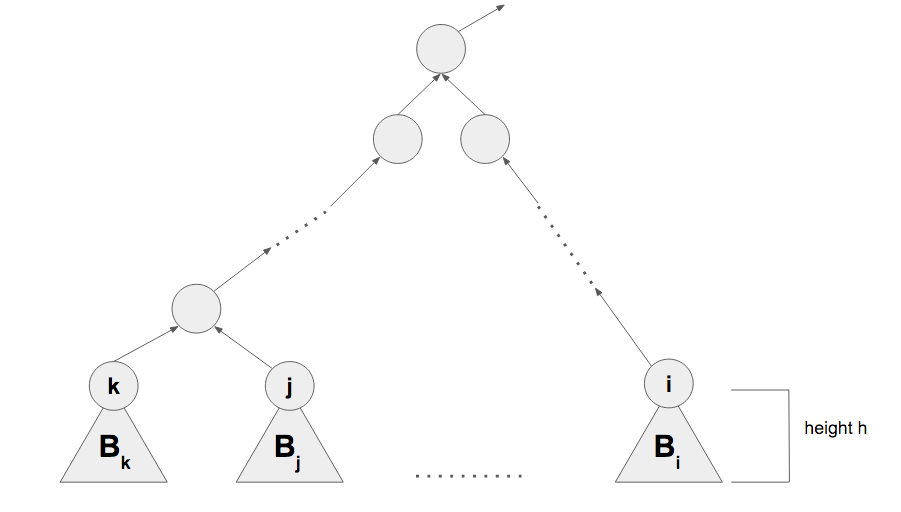}
\caption{$\sigma_i(B)$.}
\label{fig:goodinstance1}
\end{figure*}

We now describe the second map, $\sigma_j$. To construct $\sigma_j(B)$, begin by swapping the subtrees $B_j$ and $B_k$ (see Figure ~\ref{fig:goodinstance2}). Note that the bracket formed in this way is good; since we chose $k$ to be the latest player on $j$'s path to victory that can beat $j$, if $j$ meets $i$, $j$ will also beat all subsequent players and win the tournament (note that it is of course possible that $j$ doesn't even make it to $i$, in which case $\sigma_j(B)$ is still good. But it is clear that \emph{if} $j$ meets $i$, then $j$ will win the tournament, so $\sigma_j(B)$ is good in either case). Unfortunately, this map as stated is not injective; in particular, we cannot recover the height $h$ to undo the swap as in the previous case.

{
The only reason we cannot uniquely identify $k$ in the same way as when we invert $\sigma_i$ is that $i$ might meet some player $k'$ at height $h' < h$ in $B_i$ who also could beat $j$. So, intuitively, we would like to swap such players out with players who lose to $j$. Since $j$ beats all of its opponents in $B_j$, $B_j$ is an ample source of such players. We will therefore perform some additional `subswap' operations, swapping subtrees of $B_j$ and $B_k$ so as to uniquely identify $k$ as the first player $i$ meets in $\sigma_j(B)$ who can beat $j$.

Specifically, for $0 \leq h' < h$, let $a(h')$ be the opponent $i$ plays at height $h'$ in $B_i$, and let $B_i(h')$ be the subtree of $B_i$ with root $a(h')$ (note that the player that $i$ meets at height $h'$ is the root of a subtree of height $h'-1$, and that all these subtrees are disjoint). Similarly, let $b(h')$ be the opponent $j$ plays at height $h'$ in $B_j$, and let $B_j(h')$ be the subtree of $B_j$ with root $b(h')$. To construct $\sigma_j(B)$ from $B$, first swap $B_j$ and $B_k$. Then for each $h' \in [0, h)$ such that $a(h')$ would beat $j$, swap the subtrees $B_i(h')$ and $B_j(h')$. See Figure~\ref{fig:specialswap} for an illustration of a subswap operation.}


Note that $\sigma_j(B)$ is still good; it is still the case that if $j$ meets $i$, $j$ will beat all subsequent players (all we have done in that part of the bracket is perhaps alter whether or not $j$ will indeed meet $i$). On the other hand, since $j$ makes it to height $h$ in $B_j$, $j$ can beat player $b(h')$ for all $h'$, so $k$ is now the first player $i$ would encounter in $\sigma_j(B)$ who can beat $j$. From this, we can recover $k$ and thus $h$, and undo the swap of $B_i$ and $B_j$. To undo the subswaps, observe that because we started with a bad bracket $B$, that $j$ must have beaten all opponents it faces in the first $h$ rounds. Since all opponents on $j$'s path who beat $j$ at height less than $h$ were necessarily put there by our subswap operations, we can just find all such opponents and swap them back out. This process inverts $\sigma_j$, thus proving that $\sigma_j$ is injective.

\begin{figure*}
\centering
\includegraphics[scale = 0.25]{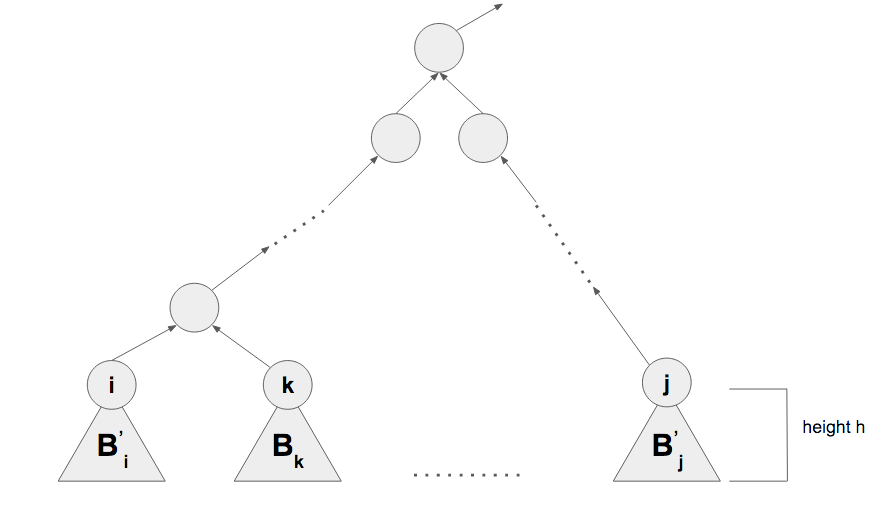}
\caption{$\sigma_j(B)$.}
\label{fig:goodinstance2}
\end{figure*}

\begin{figure*}
\centering
\includegraphics[scale = 0.25]{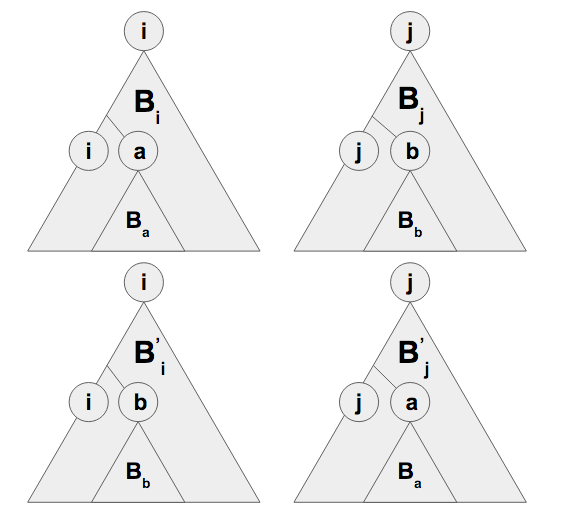}
\caption{Subswap operation for $\sigma_j$.}
\label{fig:specialswap}
\end{figure*}

Finally, note that in $\sigma_i(B)$, $k$ must play $j$ before either plays $i$, whereas in $\sigma_j(B)$, $k$ must play $i$ before either plays $j$. Therefore the ranges of $\sigma_i$ and $\sigma_j$ are disjoint, and this completes the proof.

For the reader aiming to understand our coupling argument better, Appendix~\ref{app:coupling} contains some specific examples.
\end{proof}

\subsection{Extension to randomized outcomes}

Thus far we have been assuming that all match results are deterministic and known to the players in advance. Of course, this is not true in general; in real life, the outcomes of games are inherently unpredictable. It is perhaps imaginable that this unpredictability could increase the incentive to manipulate. In this section we show that this is not the case; a simple application of linearity of expectation shows that results about deterministic tournaments still hold for their randomized counterparts. We begin by defining a randomized tournament as follows.

\begin{definition}
A \textit{randomized tournament} $\mathcal{T}$ is a random variable whose values range over (deterministic) tournaments $T$. As shorthand, we will write $\pr(T)$ to represent the probability that $\mathcal{T} = T$.
\end{definition}

Note that this definition accounts for the most straightforward generalization of tournament outcomes from deterministic to randomized, where for each match between players $i$ and $j$ we assign a probability $p_{ij}$ to the probability that $i$ beats $j$. This definition further allows for the possibility of correlation between matches (e.g., with some probability player $i$ has a good day and wins all his matches, and with some probability he has a bad day and loses all his matches).

Manipulations in this randomized model are similar to manipulations in the deterministic model in that they effectively force the result of a match to a win or a loss. Formally, let $\sigma_{ij}(T)$ for a (deterministic) tournament $T$ be the tournament formed by $T$ but where $i$ beats $j$ (if $i$ beats $j$ in $T$, then $\sigma_{ij}(T) = T$). A tournament rule $r$ is $2$-SNM-$\alpha$ if for all $i$ and $j$,

\begin{equation}\label{eq:randsnm}
\mathbb{E}_{\mathcal{T}}\left[r_{i}(\sigma_{ij}(\mathcal{T})) + r_{j}(\sigma_{ij}(\mathcal{T})) - r_{i}(\mathcal{T}) - r_j(\mathcal{T})\right] \leq \alpha 
\end{equation}

We then have the following theorem:

\begin{theorem}
If a rule $r$ is $2$-SNM-$\alpha$ in the deterministic tournament model, it is also $2$-SNM-$\alpha$ in the randomized tournament model.
\end{theorem}
\begin{proof}
Note that we can write the expectation in Equation \ref{eq:randsnm} as 

\begin{equation*}
\sum_{T}\pr(T)\left(r_{i}(\sigma_{ij}(T)) + r_{j}(\sigma_{ij}(T)) - r_{i}(T) - r_j(T)\right)
\end{equation*}

If $r$ is 2-SNM-$\alpha$ for deterministic tournaments, then each term in this sum is at most $\pr(T)\alpha$. It follows that this sum is at most $\alpha$, and therefore $r$ is also 2-SNM-$\alpha$ for randomized tournaments. 
\end{proof}

It is straightforward to generalize the above definitions and result to the case of $k$-SNM-$\alpha$.

\subsection{Other tournament formats}\label{sect:otherrules}

Finally, there are many other tournament formats that are either used in practice or have been previously studied. In this section we show that many of these formats are more susceptible to manipulation than the random single elimination bracket rule; in particular, all of the following formats are at best $2$-SNM-$1/2$. 

By far the most common tournament rule for round robin tournaments is some variant of a `scoring' rule, where the winner is the player who has won the most games (with ties broken in some fashion if multiple players have won the same maximum number of games). In voting theory, this rule is often called Copeland's rule, or Copeland's method~\cite{Copeland}.

\begin{definition}
A tournament rule $r$ is a \textit{Copeland rule} if the winner is always selected from the set of players with the maximum number of wins.
\end{definition}

We begin by showing that no Copeland rule can be 2-SNM-$\alpha$ for any $\alpha < 1$ (regardless of how the rule breaks ties).

\begin{theorem}
There is no Copeland rule on $n$ players that is 2-SNM-$\alpha$ for $\alpha < 1 - \frac{2}{n-1}$. 
\end{theorem}
\begin{proof}
Assume to begin that $n = 2k+1$ is odd, and let $r$ be a Copeland rule on $n$ players. Let $T$ be the tournament where each player $i$ beats the $k$ players $\{i+1, i+2, \dots, i+k\}$ but loses to the $k$ players $\{i-1, i-2, \dots, i-k\}$, with indices taken modulo $n$ (similar to the tournament in the proof of Theorem ~\ref{thm:lbndK}). 

Since $\sum_{i=1}^{n} r_i(T) = 1$, there must be some $i$ such that $r_{i-1}(T) + r_{i}(T) \leq \frac{2}{n}$. On the other hand, if player $i-1$ throws their match to player $i$, then player $i$ becomes the unique Copeland winner (winning $k+1$ games) and $r_{i}(T') = 1$. It follows that, for such a rule, if $r$ is 2-SNM-$\alpha$, then $\alpha \geq 1 - \frac{2}{n}$. 

If $n$ is even, then we can embed the above example for $n-1$ by assigning one player to be a dummy player that loses to all teams. This immediately implies $\alpha \geq 1 - \frac{2}{n-1}$ in this case.
\end{proof}

In~\cite{AltmanKleinberg}, Altman and Kleinberg provide three examples of tournament rules that are Condorcet-consistent and monotone: the top cycle rule, the iterative Condorcet rule, and the randomized voting caterpillar rule. We prove lower bounds on $\alpha$ for each of these in turn. Interestingly, the same tournament provides all three lower bounds.

\begin{definition} The \emph{superman-kryptonite} tournament on $n$ players has $i$ beat $j$ whenever $i < j$, except that player $n$ beats player $1$. That is, player $1$ beats everyone except for player $n$, who loses to everyone except for player $1$. \end{definition}

Now we show that the superman-kryptonite tournament provides lower bounds against the tournament rules considered in~\cite{AltmanKleinberg}. 

\begin{definition}
The \textit{top cycle} of a tournament $T$ is the minimal set of players who beat all other players. The \textit{top cycle rule} is a tournament rule which assigns the winner to be a uniformly random element of this set.
\end{definition}

\begin{theorem}\label{thm:TCR}
The top cycle rule on $n$ players is not 2-SNM-$\alpha$ for any $\alpha < 1 - \frac{2}{n}$. 
\end{theorem}
\begin{proof}
Let $T$ be the superman-kryptonite tournament on $n$ players. The top cycle in $T$ contains all the players, so $r_1(T) + r_n(T) = \frac{2}{n}$. However, if player $n$ throws their match to player $1$, player $1$ becomes a Condorcet winner and $r_1(T') = 1$. It follows that $\alpha \geq 1 - \frac{2}{n}$. 
\end{proof}

\begin{definition}
The \textit{iterative Condorcet rule} is a tournament rule that uniformly removes players at random until there is a Condorcet winner, and then assigns that player to be the winner. 
\end{definition}
\begin{theorem}\label{thm:ICR}
The iterative Condorcet rule on $n$ players is not 2-SNM-$\alpha$ for any $\alpha < \frac{1}{2} - \frac{1}{n(n-1)}$.
\end{theorem}
\begin{proof}
Let $T$ be the superman-kryptonite tournament on $n$ players. Note that no Condorcet winner will appear until either player $1$ is removed, player $n$ is removed, or all other $n-2$ players are removed. If all the other $n-2$ players are removed before players $1$ or $n$ (which occurs with probability $\frac{2}{n(n-1)}$), then player $n$ wins. If this does not happen and player $n$ is removed before player $1$ (which occurs with probability $\frac{1}{2}\left(1 - \frac{2}{n(n-1)}\right) = \frac{1}{2} - \frac{1}{n(n-1)}$), then player $1$ becomes the Condorcet winner and wins. Otherwise, player $1$ will be removed before player $n$, while some players in $2$ through $n-1$ remain, and one of them will become the Condorcet winner (the remaining player in $\{2,\ldots,n-1\}$ with lowest index). It follows that $r_1(T) = \frac{1}{2} - \frac{1}{n(n-1)}$ and $r_n(T) = \frac{2}{n(n-1)}$, so $r_1(T) + r_n(T) = \frac{1}{2} + \frac{1}{n(n-1)}$.

On the other hand, if player $n$ throws their match to player $1$, then again player $1$ becomes a Condorcet winner and $r_1(T') = 1$. It follows that $\alpha \geq \frac{1}{2} - \frac{1}{n(n-1)}$.
\end{proof}

\begin{definition}
The \textit{randomized voting caterpillar rule} is a tournament rule which chooses a winner as follows. Choose a random permutation $\pi$ of $[n]$. Start by matching $\pi(1)$ and $\pi(2)$, and choose a winner according to $T$. Then for all $i \geq 3$ match $\pi(i)$ with the winner of the most recent match. The player that wins the last match (against $\pi(n)$) is declared the winner. 
\end{definition}

\begin{theorem}\label{thm:RVC}
The randomized voting caterpillar rule on $n$ players is not 2-SNM-$\alpha$ for any $\alpha < \frac{1}{2} - \frac{n-3}{n(n-1)}$. 
\end{theorem}
\begin{proof}
Let $T$ be the superman-kryptonite tournament on $n$ players. The only way player $1$ loses is if either player $n$ occurs later in $\pi$ than player $1$ (which happens with probability $\frac{1}{2}$) or if $\pi(n) = 1$ and $\pi(1) = 2$ and they play in the first round (which happens with probability $\frac{1}{n(n-1)}$). The only way player $n$ can win is if $\pi(n) = n$ (i.e., they only play the very last game), in which case they will play player $1$ and win (this happens with probability $\frac{1}{n}$). It follows that $r_{1}(T) = \frac{1}{2} - \frac{1}{n(n-1)}$ and $r_{n}(T) = \frac{1}{n}$, so $r_1(T) + r_n(T) = \frac{1}{2} + \frac{n-2}{n(n-1)}$.

On the other hand, if player $n$ throws their match to player $1$, then again player $1$ becomes a Condorcet winner and $r_1(T') = 1$. It follows that $\alpha \geq \frac{1}{2} - \frac{n-2}{n(n-1)}$.
\end{proof}

\bibliographystyle{ACM-Reference-Format-Journals}
\bibliography{acmsmall-sample-bibfile.bib}
\appendix
\section{More Details on our Coupling Argument}\label{app:coupling}

In this appendix we present examples of bracket transformations. Recall that our transformations took as input any ``bad'' bracket, where player $i$ eventually meets player $j$, \emph{and} player $j$ will lose to some player $k$ in the future if she advances past $i$ (and $k$ is the latest such player). The players benefit from manipulating these brackets. We transformed them into ``good'' brackets, where either player $j$ is eliminated before even meeting player $i$, or where player $j$ would be the champion conditioned on getting past $i$. The players have no incentive to manipulate these brackets. 

We designed two injective transformations with disjoint images, $\sigma_i$ and $\sigma_j$. $\sigma_i$ was more straight-forward, but we include an example below anyway. $\sigma_j$ was more complex. We include below an example showing that the complexity is necessary, and then an example of $\sigma_j$. All figures are at the end.

\subsection{Example of the transformation $\sigma_i(B)$.}
Recall that $\sigma_i$ essentially swaps the sub-brackets rooted at $i$ and $k$. See Section~\ref{sec:main} for a formal description.

Consider the partial bracket $B_1$ shown in Figure~\ref{fig:normalbracket}. Then, applying the transformation $\sigma_i (B_1)$ as described in our paper will yield the bracket $B_1'$ shown in Figure~\ref{fig:sigmaibracket}. Note that this mapping is injective: by examining $\sigma_i(B)$, we see exactly where $j$ is eliminated, and conclude that this must be where $i$ met $j$ in the original $B$.

\subsection{Counterexample to a naive $\sigma_j(B)$.}
We could try using the same ideas in $\sigma_i$ for $\sigma_j$: simply swap the subtrees rooted at $k$ and $j$. Unfortunately, this mapping is not injective.

Consider the two brackets $B_3, B_4$ shown in Figure~\ref{fig:badbrackets}. Then applying this naive transformation will map these brackets to the same bracket (see Figure~\ref{fig:badbracket}), showing that the mapping may not be injective. This motivates the need for the more involved transformation $\sigma_j$ from Section~\ref{sec:main}.

Specifically, observe that in $B_3$, $i$ meets $j$ in round 2, so the depth-2 subtree rooted at $k$ would get swapped with the depth-2 subtree rooted at $j$. In $B_4$, $i$ meets $j$ in round 1, so the single node $i_1$ would get swapped with the single node $j$. It is easy, but tedious, to complete this into a full tournament/bracket. 

\subsection{Example of the transformation $\sigma_j(B)$.}
Essentially, the problem with the naive transformation is that it's hard to recover where $i$ met $j$ in the original $B$ just from the naive $\sigma_j(B)$. This is because maybe on its path to $j$, $i$ met many other competitors who also would have beaten $j$, in addition to the $k$ we swap in from the mapping. Our more involved transformation fixes this by additionally swapping all such competitors out of the subtree below $i$, so we can again recover where $i$ met $j$ in the original $B$. 

Consider the partial bracket $B_2$ shown in Figure~\ref{fig:normalbracket2} and assume that in the tournament in case $i_2$ would beat $j$. Then, applying the transformation $\sigma_j (B_2)$ as described in our paper will yield the bracket $B_2'$ shown in \ref{fig:sigmaj}. 

Note that this mapping is injective! First, we can recover where $i$ met $j$ in the original $B$ by looking at where $i$ first encounters someone who would beat $j$ in $\sigma_j(B)$. Once we learn this, we also know that in the original $B$, $j$ actually advanced this far in the tournament to meet $i$, so we know exactly which subtrees we need to un-swap with subtrees of $i$.

\begin{figure}[htb]
\includegraphics[scale = 0.4]{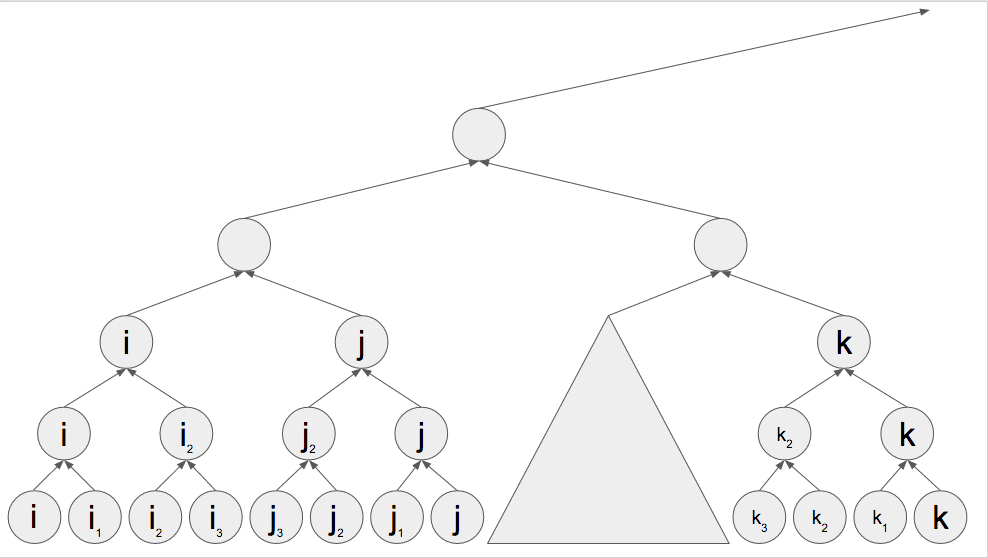}
\caption{A partial bracket $B_1$.}
\label{fig:normalbracket}
\end{figure}  

\begin{figure}[htb]
\includegraphics[scale = 0.4]{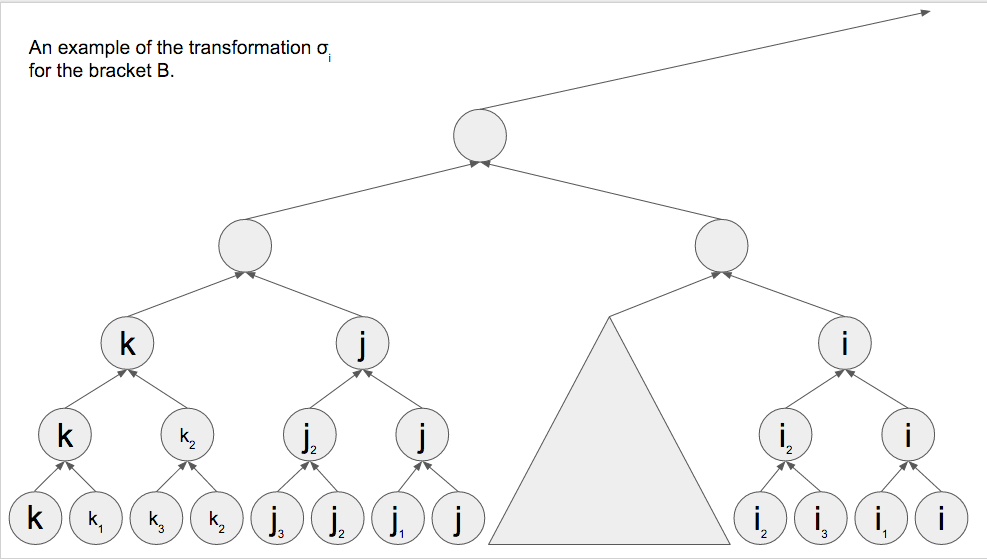}
\caption{$\sigma_i(B_1)$.}
\label{fig:sigmaibracket}
\end{figure}  

\begin{figure}[htb]
\includegraphics[scale = 0.4]{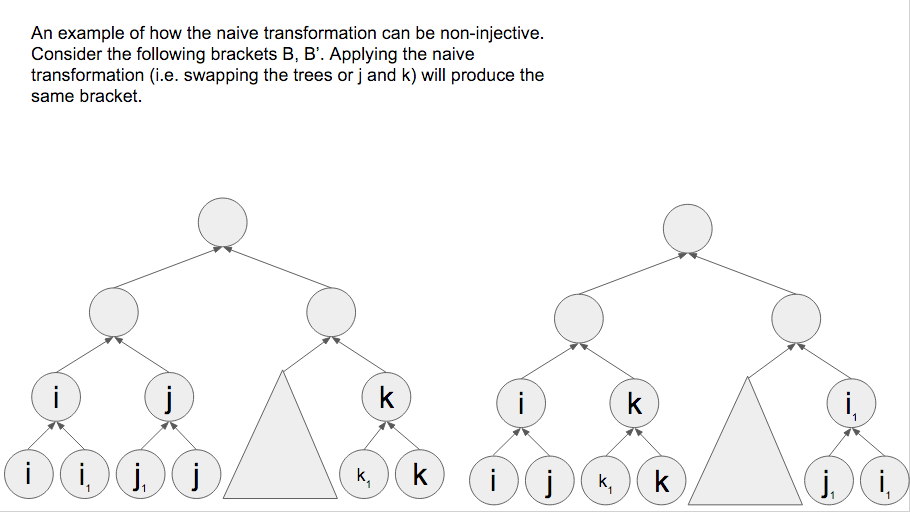}
\caption{Two partial brackets $B_3$, $B_4$.}
\label{fig:badbrackets}
\end{figure}  

\begin{figure}[htb]
\includegraphics[scale = 0.4]{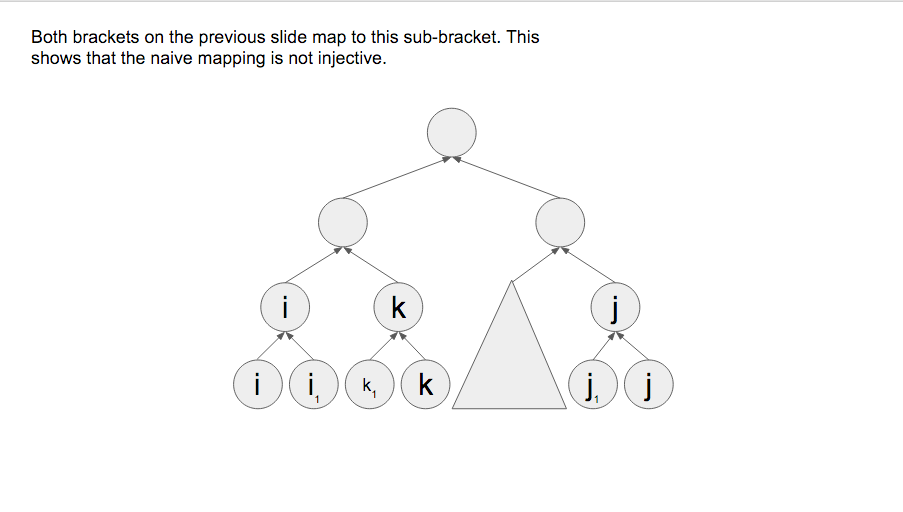}
\caption{Swapping the subtrees corresponding to $j,k$ in both brackets above yields this bracket.}
\label{fig:badbracket}
\end{figure}  

\begin{figure}[htb]
\includegraphics[scale = 0.4]{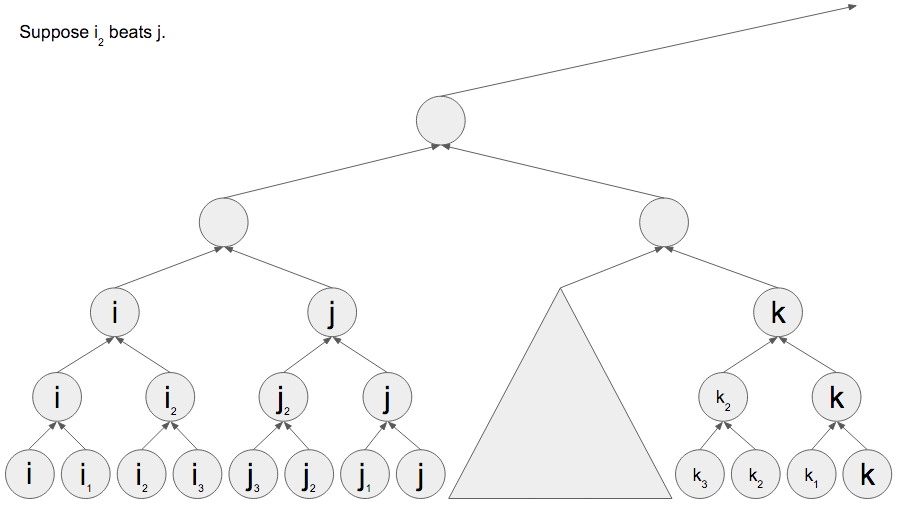}
\caption{A partial bracket $B_2$.}
\label{fig:normalbracket2}
\end{figure}  

\begin{figure}[htb]
\includegraphics[scale = 0.4]{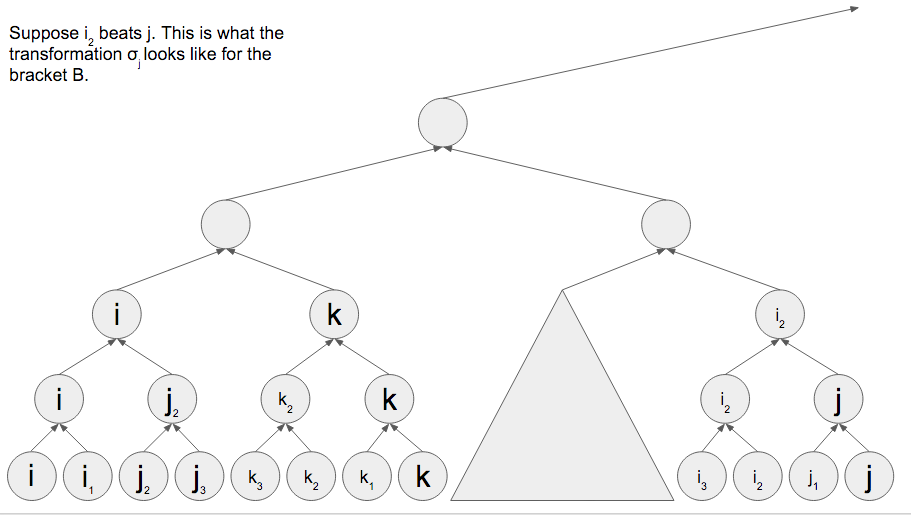}
\caption{$\sigma_j(B_2)$.}
\label{fig:sigmaj}
\end{figure}

\end{document}